 \documentclass[a4paper, 12pt]{article}
\usepackage[english]{babel}
\usepackage{amssymb}
\usepackage{amsfonts}

\usepackage{bbm}
\usepackage{amsmath}
\usepackage{dsfont}

\usepackage{float}

\usepackage{psfrag,epsf}
\usepackage{graphics}
\usepackage{graphicx}
\usepackage{amsmath}
\usepackage{color}
\usepackage{rotating}
\usepackage{varioref}
\usepackage{multirow}
\usepackage{pdflscape}
\usepackage{natbib}
\usepackage{hyperref}	

\newtheorem{theorem}{Theorem}

\newtheorem{lemma}[theorem]{Lemma}

\newtheorem{proposition}[theorem]{Proposition}

\newenvironment{proof}[1][Proof]{\noindent\textbf{#1.} }{\ \rule{0.5em}{0.5em}}

\DeclareMathSymbol{\leqslant}{\mathalpha}{AMSa}{"36} 
\DeclareMathSymbol{\geqslant}{\mathalpha}{AMSa}{"3E} 
\DeclareMathSymbol{\eset}{\mathalpha}{AMSb}{"3F}     

\newcommand{\R}{\mathbb{R}}

\newcommand{\PEfont}{\mathrm}

\newcommand{\E}{\ensuremath{\PEfont E}}

\renewcommand{\epsilon}{\varepsilon} 
\renewcommand{\theta}{\vartheta} 
\renewcommand{\rho}{\varrho} 

\usepackage[left=1in, right=1in, top=1.3in, bottom=1.2in]{geometry}
\usepackage{setspace} 
\begin{document}
\begin{spacing}{1.3}

\title{\bf Backtesting Lambda Value at Risk}%

  \author{Jacopo Corbetta\\
    \small CERMICS, \'{E}cole des Ponts , UPE, Champs sur Marne, France.\\
   \small Zeliade Systems, 56 rue Jean-Jacques Rousseau, Paris, France.\\
    and \\
    Ilaria Peri \\
   \small Department of Economics, Mathematics and Statistics, Birkbeck University of London, England}
  \maketitle

\begin{abstract}
{\small 
A new risk measure, Lambda value at risk ($\Lambda VaR$), has been recently proposed as a generalization of Value at risk ($VaR$). $\Lambda VaR$ appears attractive for its potential ability to solve several problems of $VaR$. This paper provides the first study on the backtesting of $\Lambda VaR$.  
We propose three nonparametric tests which exploit different features. Two tests are based on simple results of probability theory. One test is unilateral and is more suitable for small samples of observations. A second test is bilateral and provides an asymptotic result. A third test is based on simulations and allows for a more accurate comparison among $\Lambda VaRs$ computed with different assumptions on the asset return distribution. 
Finally, we perform a backtesting exercise that confirms a higher performance of $\Lambda VaR$ in respect to $VaR$ especially when it is estimated with distributions that better capture tail behaviour. 
 } 
\bigskip
\end{abstract}



\end{spacing}

\thispagestyle{empty}

\pagebreak

\section{Introduction}

Risk measurement and its backtesting are matter of primary concern to financial industry. 
Value at risk ($VaR$) has become the most widely used risk measure. Despite its popularity, after 
the recent financial crisis, $VaR$ has been extensively criticized by academics and risk managers. 
Among these critics, we recall the inability to capture the tail risk and the lack of reactivity 
to market fluctuations. Thus, the suggestion of the Basel Committee, in the consultative 
document \cite{FRTB}, is to consider alternative risk measures that are able to overcome the $VaR$'s weaknesses.

A new risk measure, Lambda Value at Risk ($\Lambda VaR$), has been introduced by a theoretical point 
of view by \cite{FMP}. $\Lambda VaR$ is a generalization of the $VaR$ at confidence level $\lambda$. 
Specifically, $\Lambda VaR$ considers a function $\Lambda$ instead of a constant confidence level $\lambda$, 
where $\Lambda$ is a function of the losses. Formally, given a monotone and right continuous function 
$\Lambda: \R \rightarrow (0, 1)$, the $\Lambda VaR$ of the asset return $X$ is a map that associates to its 
cumulative distribution function $F(x) = P(X \leq x) $ the number:
\begin{equation}
\Lambda VaR=-\inf \left\{ x\in \mathbb{R} \mid F(x)>\Lambda (x)\right\}\,.
\label{eq LVaR}
\end{equation}
This new risk measure appears to be attractive for its potential ability to solve several problems of $VaR$. 
First of all, it seems to be flexible enough to discriminate the risk among return distributions with different 
tail behavior, by assigning more risk to heavy-tailed return distributions and less in the opposite case. 
In addition, $\Lambda VaR$ may allow for a rapid changing of the interval of confidence when the market conditions change.

Recently, \cite{HP15} proposed a methodology for computing $\Lambda VaR$ and a first attempt of backtesting based on the hypothesis testing framework by \cite{K95}. In this study, the accuracy of the $\Lambda VaR$ model is evaluated
by considering the maximum of the $\Lambda$ function as confidence level. However, the level of coverage provided by the $\Lambda VaR$ model may not be constant at any time; hence, this method misses to assess the actual $\Lambda VaR$ performance.

The objective of this paper is to propose the first theoretical framework for the backtesting of $\Lambda VaR$. 
We present three backtesting methodologies which exploit different features and may be used with different aims. Our tests evaluate if $\Lambda VaR$ provides an accurate level of coverage, this means that the 
probability of a violation occurring \textit{ex-post} actually coincides with the one predicted by the model. In respect to the hypothesis test proposed in \cite{HP15}, we consider a null hypothesis which better evaluates the benefits introduced by the $\Lambda VaR$ flexibility. Our tests can be easily extended to $VaR$ allowing for a proper comparison among the two risk measures.

Two of these tests are based on simple test statistics whose distribution is obtained by applying results of probability theory. The first test is unilateral and 
provides more precise results for shorter backtesting time windows (e.g. 250 observations). The second test is bilateral 
and provides an asymptotic result that makes it more suitable for larger samples of observations. 

We propose a third test that is inspired to the approach used by \cite{AS14} for the Expected Shortfall backtesting. Here, the distribution of the test statistic is obtained by Monte Carlo simulations. This test allows to better evaluate the impact of the assumption on the model generating data and compare different choices on the asset return distribution. 

Finally, we conduct an empirical analysis where we experiment and compare the results of our backtesting proposals for $\Lambda VaR$, computed using the same dynamic benchmark approach proposed by \cite{HP15}. The backtesting
exercise has been performed along six different time windows throughout all the global financial crisis (2006-2011). 

The paper is structured as follows: Section \ref{sec:models} presents the $VaR$ and $\Lambda VaR$ models; Section \ref{backtest} introduces our backtesting proposals; Section \ref{sec:secempirical} describes and shows the results of the empirical analysis; Appendix collects the proofs.

\section{$VaR$ and $\Lambda VaR$ models}\label{sec:models}

Let us consider a probability space $(\Omega, (\mathcal{F}_t)_T, \mathbb{P}_t)$, 
where the sigma algebra $\mathcal{F}_t$ represents the information at time $t$. We 
assume that $X$ is the random variable of the returns of an asset distributed 
along a real (unknown) distribution $F_t$, i.e. $F_t(x):=\mathbb{P}_t(X_t<x)$, 
and it is forecasted by a model predictive distribution $P_t$ conditional to 
previous information, i.e. $P_t(x) = \mathbb{P}_t(X_t\leq x|\mathcal{F}_{t-1})$. 

We can measure the risk of the asset return $X$ using the classical $VaR$, by 
attributing to $X$ at time $t$ the following value:
\begin{equation}
VaR_t=-\inf \left\{ x\in \mathbb{R} \mid P_t(x)>\lambda\right\}\,.
\label{eq VaR}
\end{equation}
The objective of this study is the alternative risk measure proposed by \cite{FMP}, $\Lambda VaR$, that attributes 
to $X$ at time $t$ the following value:
\begin{equation}
\Lambda VaR_t=-\inf \left\{ x\in \mathbb{R} \mid P_t(x)>\Lambda_t (x)\right\}\,.
\label{eq LVaR2}
\end{equation}
where $\Lambda_t$ is a monotone function that maps $x \in \mathbb{R}$ in 
$(\lambda_m,\lambda_M)$ with $\lambda_{m}>0$ and $\lambda_{M}<1$. When $\Lambda_t$ is constant and equal to $\lambda \in (0,1)$ for any $x$, $\Lambda VaR$ coincides with $VaR$ at confidence level $\lambda$. The interesting feature of $\Lambda VaR$ is the sensitivity to tail risk, in particular, it is able to discriminate the risk of assets having the same $VaR$ at some level $\lambda$ but different tail behavior. Thus, $\Lambda VaR$ may allow to enhance the capital requirement in case of expected greater losses.

\cite{HP15} proposed a method to compute the $\Lambda$ function that is called 
dynamic benchmark approach. Here, the $\Lambda$ function is taken as  proxy 
of the tails of the market return distribution. This feature allows $\Lambda VaR$ to assess the different asset reactions  in respect to the market by detecting different confidence levels. This approach is also dynamic 
since $\Lambda$ is re-estimated at each time $t$ according the information in 
$t-1$. In this way, $\Lambda VaR$ incorporates the recent market fluctuations and adjusts the confidence level according the different asset reactions.   

The authors proposed different models to compute $\Lambda VaR$. One proposal is to obtain $\Lambda$ by linear interpolation of $n$ points $(\pi_{i},\lambda_{i})$ for any $\pi_{1} \leq x < 
\pi_{n}$ and fix $\Lambda$ constantly equal to the lower (upper) bound for any $x \leq \pi_{1}$ and to the upper (lower) bound for any $x \geq \pi_{n}$ in the increasing (decreasing) case. In their empirical analysis, the authors chose 4 points ($n=4$). In particular, on the probability axis, they set the $
\Lambda$ lower bound $\lambda_{m}=0.001$, the upper bound $\lambda_{M}=0.01$ and the 
others $\lambda_{i}$ values, with $i=2,..,3$, by an equipartition of the interval 
$(0,\lambda_{M}]$. On the losses axis, they fix $4$ points $\pi_{i}$ equal to order statistics of the return distribution of some selected market 
benchmarks. Specifically, $\pi_{1}$ is equal to the minimum of all benchmark 
returns: $\pi_{1}=\min x_{t,j}$ where $x_{t,j}$ is the realized return of the 
$j$-th benchmark, for $t=1,..,T$ and $T$ is the time horizon (i.e. number of days in the rolling window), and for $j=1,\ldots,B$ and $B$ is the number of benchmarks; 
$\pi_{2}$, $\pi_{3}$, and $\pi_{4}$ are equal to the maximum, mean, and minimum of 
the benchmarks' $\lambda \%$-$VaR$, respectively.

In the next section, we will recall the first attempt of backtesting for $\Lambda VaR$, explain its limit and introduce our hypothesis test proposals.  

\section{$VaR$ and $\Lambda VaR$ backtesting models}
\label{backtest}
The \cite{BIS96b} refers to backtesting as the process of "comparing daily profits
and losses with model-generated risk measures to gauge the quality and accuracy of risk
measurement systems". A violation occurs when the risk measure estimate is not able to cover the realized return (profit and loss, P$\&$L). In the same Basel 2 Accord, the Commitee has also set up the first regulatory backtesting framework for the $VaR$ measure, known as traffic light approach. This procedure monitors the 1$\%$ $VaR$ violations over the last 250 days. Afterwards, many alternative proposals have been introduced in the literature for $VaR$; we refer to \cite{C05}, \cite{C10}, and \cite{BCP11} for a detailed review.

Let us denote with $x_t$ the realization of the asset return $X$ at time $t$. In 
order to perform the backtesting of a risk measure, we need to construct the sequence of random variables representing the violations, 
$\{I_t\}^T_{t=1}$, across $T$ days, as follows:
\begin{equation}
 I_t=\begin{cases}
      1 \quad \text{if } x_t<y_t\\
      0 \quad \text{otherwise } 
     \end{cases}
\end{equation}

where $y_t$ is the return forecasted by the risk measure. The hit sequence is 
equal to 1 on day $t$ if the realized returns on that day, $x_t$, 
is smaller than the value $y_t$ predicted by the risk measure at time $t-1$ for the day $t$, i.e. $\Lambda VaR_t$ or $VaR_t$. 
If $y_t$ is not exceeded (or violated), then the hit sequence returns a 0. We observe that $I_t$ is a random variable that follows a Bernoulli distribution, that is:
\begin{equation}
I_t \sim B(\lambda_t)
\end{equation}
where $\lambda_t$ is the probability of having an exception at time $t$. 

In the following, we focus on testing the unconditional coverage property of the risk measures that assumes the independence of the violations $I_t$. A common practice in the industry is testing the independence by visual inspection of the cluster of the exceptions (see \cite{AS14}, section 1). We conduct an empirical analysis with the available data which shows that $\Lambda VaR$ clusters the exceptions considerably less than $VaR$ and suggests a higher level of independence of the $\Lambda VaR$ exceptions (as shown in Figure \eqref{figure1}). The reason behind might be that $\Lambda VaR$ is recalculated at each time $t$ incorporating the recent market movements and, in this way, it may avoid sequential violations. However, in order to have a complete assessment of the accuracy of a risk measure, a specific test of independence is required. In the case of $\Lambda VaR$, one cannot rely on the immediate extension of the $VaR$ framework since the exceptions are not identically distributed. This requires a more complex analysis that we leave for a 
future study.

\begin{figure}[h!]

\centering
 \makebox[\textwidth][c]{
\includegraphics[scale=0.28]{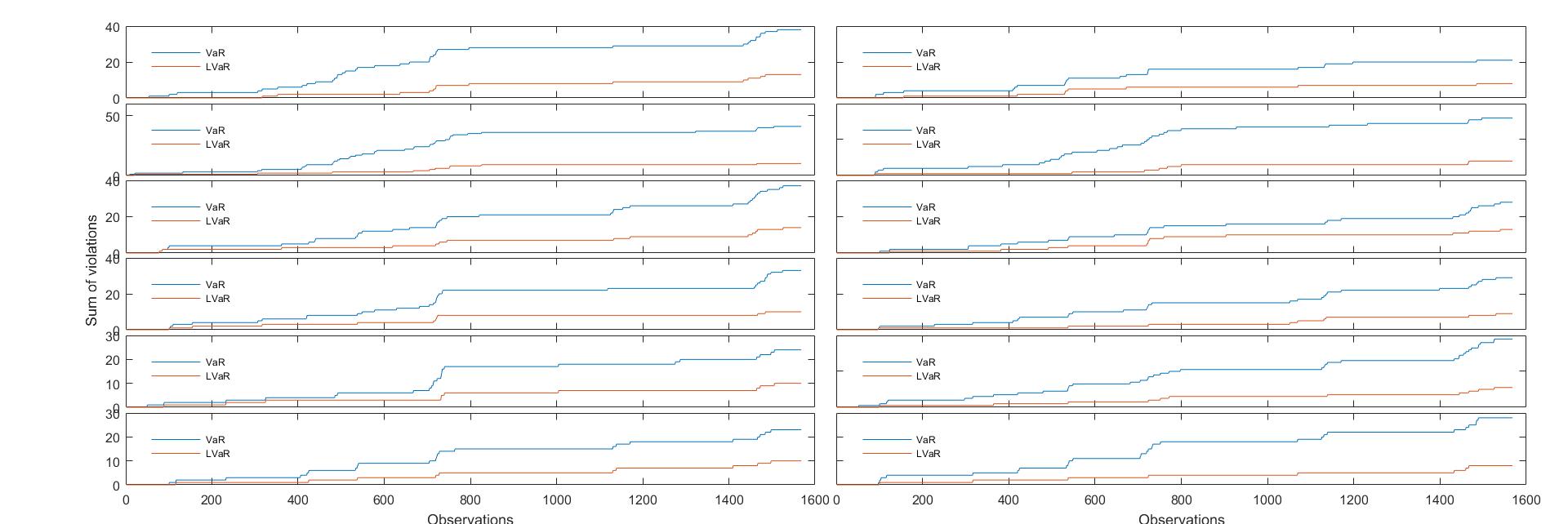}}
 
 \caption{\footnotesize{Time evolution of the sum of violations for $1 \% VaR$ and $\Lambda VaR$. The table shows the evolution over the global financial crisis of the sum of violations of the  $1 \% VaR$ and the increasing $\Lambda VaR$ model.}}  
\label{figure1}

\end{figure}

The first theoretical proposal for the backtesting of $VaR$ is given by \cite{K95}, where the author considers the following null and alternative 
hypothesis:
\begin{equation}\label{eq: ipotesi perihitaj}
\begin{split}
H^{K}_0 &: \lambda_t \leq (=) \lambda^0  \ \ \text{ for any $t$}\\
H^{K}_1 &: \lambda_t > \lambda^0 \ \ \text{ for some $t$ and equal otherwise}
\end{split}
\end{equation} 

where $\lambda^0$ is the $VaR$ confidence level. The $VaR$ at level $\lambda^0$ is accepted if the frequency of the exceptions does not exceed the confidence level $\lambda^0$ for any $t$.

Recently, \cite{HP15} have proposed a backtesting method for $\Lambda VaR$ by adapting the classical Kupiec test 
for $VaR$. They consider the following null and alternative 
hypothesis:
\begin{equation}\label{eq: ipotesi perihitaj1}
\begin{split}
H^{K}_0 &: \lambda_t \leq \text{max}(\Lambda) \ \  \text{ for any $t$}\\
H^{K}_1 &: \lambda_t > \text{max}(\Lambda)\ \ \text{ for some $t$ and equal otherwise}
\end{split}
\end{equation}
Substantially, $\Lambda VaR$ is accepted if the frequency of violations is less than $\text{max}(\Lambda)$. This is an unilateral hypothesis test that can be conducted by using the same log-likelihood 
ratio and critical value of the $VaR$ test. This approach permits to verify if the coverage objective given by the $\Lambda$ maximum has been reached, however, it 
does not allow to evaluate the accuracy of $\Lambda VaR$ at any time $t$. 

Indeed, if the $\Lambda VaR$ model is 
correct, at time $t$ we should be expecting that the hit sequence assumes value 1 with 
probability 
\begin{equation}
\lambda^0_t=\Lambda_t(-\Lambda VaR_t)
\end{equation}
and 0 with probability $1-\lambda^0_{t}$. This intuition is correct if both $\Lambda_t$ and $P_t$ are continuous. In case this does not occur, we have $\lambda^0_t=P_t(-\Lambda VaR_t)$.

As a consequence, the random variables $I_t$ of the violations for $\Lambda VaR$ are 
not identically distributed, which implies that usual likelihood backtesting 
framework (POF by \citealt{K95} , TUFF by \citealt{C10} etc.) cannot be directly applied.

Hence, if $\Lambda VaR$ is correct, the null hypothesis should be: 
\begin{equation}
H_0 : \lambda_t=\lambda^0_t \text{ for any $t$}
\label{null}
\end{equation}
while the alternative hypothesis, either:
\begin{equation}
H_1 : \lambda_t \neq \lambda^0_t \text{ for some $t$ }
\label{alternative1}
\end{equation}
in case of a bilateral test, or:
\begin{equation}
H_1 : \lambda_t > \lambda^0_t \text{ for some $t$ and equal otherwise
}
\label{alternative2}
\end{equation}
in case of an unilateral test where we reject in presence of risk under-estimation. 

The null hypothesis in \eqref{null} allows to evaluate if 
$\Lambda VaR$ guarantees the level of coverage predicted by the $\lambda^0_t$ parameter. In this way, we are able to assess the correctness of $\Lambda VaR$ more precisely than \cite{HP15}. Notice that a rejection of $H^K_0$ in \eqref{eq: ipotesi perihitaj1} implies a rejection of $H_0$ in \eqref{null}. Observe also that these hypothesis tests are also valid for $VaR$ at confidence level $\lambda^0$ by fixing $\lambda^0_t=\lambda^0$ for any $t$. 

In order to test the accuracy of the $\Lambda VaR$ model, we propose three test statistics. The distribution of the first two test statistics is obtained by exploiting simple results of probability theory. In particular, the second test provides an asymptotic result, hence it is more suitable for larger samples of observations (i.e. time horizon larger than 500).

We propose also a third test that is more useful to check if $\Lambda VaR$ has been estimated with the correct distribution function, $P_t$. 
Here, the correctness of the null hypothesis is evaluated by a simulation exercise.  

We suggest that the first two tests are used for an initial validation of the $\Lambda VaR$ model, while the third test is used as second step for selecting the best choice of estimation for the asset return distribution. 

\subsection{Test 1}

We set the null and the alternative hypothesis as in \eqref{null} and \eqref{alternative2}, 
respectively. We construct this first test by defining the test statistic $Z_1$ equal to the 
number of violations over the time horizon $T$, as follows:
\begin{equation}
 Z_1:=\sum_{t=1}^T I_t
\end{equation}
The distribution of $Z_1$ is obtained by applying classical results of 
probability theory. If the violations $I_t$ independently occurs, the sum of 
independent Bernoulli with different mean follows a Poisson Binomial distribution 
$(\lambda_t)$, thus we have that under $H_0$:
\begin{equation}
Z_1 \sim \text{Poiss.Bin}(\{ \lambda^0_t \}).
\end{equation}
This test is in principle a bilateral test, with critical region: 
$C=\left\{z_1: z_1<q_{Z_1}(\frac{\alpha}{2})\right\}\cup\left\{z_1:z_1 \geq q_{Z_1}(1-\frac{\alpha}{2})\right\}$, where $\alpha$ denotes the significance level of the test (i.e. 1 type error) 
and $q_{Z_1}$ is the quantile of the $Z_1$ distribution under $H_0$, i.e. $P_{Z_1}$. However, in the backtesting practice, this test can be treated as unilateral, where the 
critical region is given by:
\begin{equation}
C_{Z_1} =\left \{ z_1:z_1 \geq q_{Z_1}(1-\alpha) \right \} = \left \{ z_1: P_{Z_1}(z_1) > 1-\alpha \right \}
\end{equation}
Indeed, the probability that $z_1$ falls in the left side of the critical region $C$ is null, since $q_{Z_1}(\frac{\alpha}{2})$ is zero any time the following relation is satisfied: $(1-\max({\lambda}_t))^T>\alpha/2$. This is typical for usual test significance levels ($\alpha=10 \%$ or lower), usual time horizon $T=250$ and $1\%$-$VaR$ or $1\%$-$\Lambda VaR$ (since $\lambda_t \leq 0.01$).

This test represents an extension of the traffic light approach by \cite{BIS96b} to $\Lambda VaR$ with two bands instead of three. In particular, for $VaR$ at confidence level $\lambda^0$, under $H_0$ we have:
\[
Z_1 \sim \text{Bin}(T,\lambda^0)
\]
that is $Z_1$ follows a Binomial distribution. In the empirical analysis we fix $\alpha=10 \%$ and we compare the results with $VaR$. 

\subsection{Test 2}
We propose a second test statistic that is founded on a result of probability theory known as Lyapunov theorem. We set the null and the alternative hypothesis as in \eqref{null} and \eqref{alternative1}, respectively. We propose another test statistic defined as follows: 
\begin{equation}
Z_2:=\frac{\sum_{t=1}^T(I_t-\lambda^0_t)}{\sqrt{\sum_1^T\lambda^0_t(1-\lambda^0_t)}}
\end{equation}
Under $H_0$, $Z_2$ is asymptotically distributed as a Standard Normal, formally:
\begin{equation}
Z_2\xrightarrow{d} N(0,1)
\end{equation}
This result follows from the application of Lemma \ref{LemmaLyap} and the 
Lyapunov's theorem (see Appendix for details).

We remark that this is a bilateral test. Thus, we reject the null hypothesis $H_0$ if 
the realization $z_2$ of the test statistic stays in the following critical region: 
\begin{equation}
C_{Z_2}:=\left \{z_2:z_2(x)< q_{Z_2} \left( \frac{\alpha}{2} \right) \right\} \cup \left\{z_2:z_2(x)> q_{Z_2} \left ( 1-\frac{\alpha}{2} \right) \right\}
\end{equation}
where $\alpha$ is the significance level of the test, and $q_{Z_2}$ is the quantile 
function of the Standard Normal distribution $P_{Z_2}$.

Also for this test, in the empirical analysis, we fix $\alpha=10 \%$ and we compare the results with $VaR$. 

\subsection{Test 3}

The third test is inspired by \cite{AS14} and focused 
on another aspect. The aim of this test is to directly verify if $\Lambda VaR$ has been estimated under the correct assumption on the distribution $P_t$ of the returns. To this purpose we build a test statistic, $Z_3$, and we proceed by simulating its distribution using the same assumption as for the asset return distribution in the risk measure computation.

We set the null and the alternative hypothesis as in \eqref{null} and \eqref{alternative2}, 
respectively, and we define $Z_3$ as follows:
\begin{equation} \label{Z3}
 Z_3:=\frac{1}{T}\sum_{t=1}^T (\lambda^0_{t}-I_t)=\frac{1}{T}\sum_{t=1}^T \lambda^0_{t}- \frac{1}{T}\sum_{t=1}^T I_t
\end{equation}

We observe that under $H_0$, we have $\E[Z_3]=0$, while under $H_1$, $\E[Z_3] <0$ for $\Lambda VaR$ (see Proposition \eqref{prop} in Appendix). So, if the model is correct the realized value $z_3$ is expected to be zero. On the other hand, a negative $z_3$ is a signal that the model estimation does not allow for covering the risk. 

Under $H_0$ the distribution of $Z_3$ depends on the assumption for the distribution $P_t$ of the asset returns. Hence, we perform the test by simulating $M$ scenarios of the distribution $P_t$ of the returns at each time $t$, with $t=1, \ldots, T$. In this way, we obtain at time $T$ the 
distribution $P_{Z_3}$ of the test statistic under $H_0$. In order to construct the 
critical region we need to study the behavior of $P_{Z_3}$ when the 
distribution of the returns changes from $P$ to $F$. Let us compute $P_{Z_3}$:
\begin{displaymath}
\begin{split}
P_{Z_3}&=\mathbb{P}\left(Z_3 \leq z \right)=\mathbb{P}\left(\frac{1}{T}\sum_{t=1}^T (\lambda^0_{t}-I_t) \leq z\right)\\
&=\mathbb{P}\left( \sum_{t=1}^T(-I_t) \leq zT-\sum_{t=1}^T \lambda^0_{t}\right) \\
&=\mathbb{P}\left( \sum_{t=1}^T I_t \geq -zT+\sum_{t=1}^T \lambda^0_{t}\right) 
\end{split}
\end{displaymath} 
where $\sum_{t=1}^T I_t$ is distributed as a Binomial Poisson of parameter $\{\lambda_t\}$.
We observe that $P_{Z_3}$ is an increasing function of $\{\lambda_t\}$ (i.e. $P_{Z_3}$ shifts to left when $\lambda_t$ increases).
As a consequence,
given a significance level $\alpha$, we reject the null hypothesis when the p-value $p=P_{Z_3}(z)$ is smaller 
than $\alpha$.

In the empirical analysis we conduct $M=10000$ simulations using the same assumptions on 
the asset return distribution as for the risk measures computation. We set the test significance level $\alpha$ at $10 \%$.

This test allows to verify how the choice of the asset return distribution influences 
the risk coverage capacity of $\Lambda VaR$, that, instead, is not directly assessed 
by Test 1 and Test 2. Hence, the best use of Test 3 is comparing the results between 
the same kind of $\Lambda VaR$ models, but estimated with different assumptions on the P$\&
$L distribution (i.e. Historical, Montecarlo Normal and GARCH, etc.).  

The limit of this test is that requires a massive storage of information, since at time $T$ 
we need all the predictive distributions $P_t$ of the returns for $t=1,\ldots, T$.

\section{Empirical analysis}\label{sec:secempirical}

In this section, we provide an empirical analysis of the backtesting methods of 
$\Lambda VaR$ that we have defined in Section \eqref{backtest}. We applied the tests 
to a slightly different version of the $1\%- \Lambda VaR$ models proposed in \cite{HP15} 
and to the $1\%-VaR$ model. We compare the backtesting results with the Kupiec-type test 
proposed in \cite{HP15} for $\Lambda VaR$ and with the classical Kupiec's test for $VaR$. 

We refer to the same dataset as in \cite{HP15}, consisting in daily data of 12 stocks 
quoted in different countries along different time windows throughout the global 
financial crisis (specifically, from January 2005 to December 2011). These comprise
the stocks of Citigroup Inc. (C UN Equity) and Microsoft Corporation (MSFT UW Equity) 
for the United States, Royal Bank of Scotland Group PLC (RBS LN Equity) and Unilever 
PLC (ULVR LN Equity) for the United Kingdom, Volkswagen AG (VOW3 GY Equity) and Deutsche
Bank AG (DBK GY Equity) for Germany, Total SA (FP FP Equity) and BNP Paribas SA (BNP FP 
Equity) for France, Banco Santander SA (SAN SQ Equity) and Telefonica SA (TEF SQ Equity) 
for Spain, and Intesa Sanpaolo SPA (ISP IM Equity) and Enel SPA (ENEL IM Equity) for Italy. The market benchmarks for the $\Lambda VaR$ computation have been chosen among the market indexes with the highest volume of exchanges; these are S$\&$P500, FTSE 100, and EURO STOXX 50.

The computation of the risk measures is based on different assumptions on the distribution 
of the asset returns. We consider the classical Historical and Normal simulation approach and we add robustness to the analysis by implementing GARCH models with t-student increments and the Extreme Value Theory (EVT) method based on the generalised Pareto distribution (we remand to \cite{McNeil} for a review on this method).
The estimation of the parameters is based on 250 days of observations for the Historical and Normal assumption, while 500 days are considered for the GARCH model. For the Extreme Value Theory method, we implement an automatic routine to identify the threshold in the different time windows.   

The backtesting exercise is conducted comparing the realized ex-post daily returns with the daily $VaR$ and $\Lambda VaR$ estimates of the 12 stocks over 
the time period of 1 year. In particular, we split the analysis into six
different 2-year time windows (250 days for the risk measure computation 
and 1 year for the backtesting).

\subsection{Results}
\subsubsection{Violations and Kupiec test}

We first report the results of the violations and the Kupiec test for the $VaR$ model and 
the Kupiec-type test adapted by \cite{HP15} for the $\Lambda VaR$ model. We compute the average 
number of violations and acceptance rate over all the assets and different time horizon
$T$. The results presented, hereafter, in Table \eqref{Table1} are under the assumption of Historical distribution of the asset returns.
\begin{center}
\begin{table}[H]
 {\resizebox{1.0 \textwidth}{!}{ \renewcommand\arraystretch{2.0}  \begin{tabular}[htb] {rr |rrrrrrr|rrrrrr}
& \multicolumn{1}{c}{} & \multicolumn{6}{c}{\textit{\textbf{Average number of violations }}} &       \multicolumn{1}{c}{}      &        \multicolumn{6}{c}{\textit{\textbf{Kupiec-Test }}} \\

\hline

& \multicolumn{1}{c}{\textit{}} & \multicolumn{1}{|c}{2006} & \multicolumn{1}{c}{2007} & \multicolumn{1}{c}{2008} & \multicolumn{1}{c}{2009} & \multicolumn{1}{c}{2010} & \multicolumn{1}{c}{2011} & \multicolumn{1}{c}{} & \multicolumn{1}{|c}{2006} & \multicolumn{1}{c}{2007} & \multicolumn{1}{c}{2008} & \multicolumn{1}{c}{2009} & \multicolumn{1}{c}{2010} & \multicolumn{1}{c}{2011} \\

	
\hline

\multicolumn{1}{c|}{\textit{VaR}}& \multicolumn{1}{|c|}{\textit{1\%}} &3.42 & 5.33 & 11.58 & 0.75 & 3.08 & 6.83& 
\multicolumn{1}{c|}{\multirow{5}[10]{*}} 
& 100 \% & 83 \% & 0 \% & 100 \% & 92 \% & 50 \% \\

\hline

\begin{sideways}\end{sideways} & \textit{} &  \textbf{3.42} & \textbf{5.33} & \textbf{11.58} &\textbf{ 0.75} & \textbf{3.08} & \textbf{6.83} & \begin{sideways}\end{sideways} & \textbf{100\%} & \textbf{83\%} & \textbf{0\%} & \textbf{100\%} & \textbf{92\%} & \textbf{50\%} \\

\hline

\multicolumn{1}{c|}{\multirow{2}[4]{*}{\textit{$\Lambda VaR$ $1\%$ (decr)}}} & \multicolumn{1}{|c|}{\textit{ (VaR 5\%)}} 
& 2.25 & 3.67 & 7.00 & 0.67 & 2.00 & 4.25 & \multicolumn{1}{c|}{\multirow{2}[4]{*}{}}  & 100 \% & 83 \% & 42 \% & 100 \% & 100 \% & 83 \% \\

\multicolumn{1}{c|}{} & \multicolumn{1}{|c|}{\textit{ (VaR 1\%)}} 
& 2.17 & 2.33 & 5.75 & 0.67 & 1.58 & 4.00 & \multicolumn{1}{c|}{}  & 100 \% & 83 \% & 67 \% & 100 \% & 100 \% & 83 \% \\

\hline

\begin{sideways}\end{sideways} & \textit{} & \textbf{2.21} & \textbf{3.00} & \textbf{6.38 }& \textbf{0.67} & \textbf{1.79} & \textbf{4.13}  & \multicolumn{1}{c|}{\begin{sideways}\end{sideways}}& \textbf{100 \%} & \textbf{83 \%} & \textbf{54 \%} & \textbf{100 \%} & \textbf{100 \%} & \textbf{83 \%}\\

\hline

\multicolumn{1}{c|}{\multirow{2}[4]{*}{\textit{$\Lambda VaR$ $1\%$ (incr)}}} & \multicolumn{1}{|c|}{\textit{ (VaR 5\%)}} & 1.17 & 1.00 & 3.92 & 0.42 & 0.92 & 2.75 & \multicolumn{1}{c|}{\multirow{2}[4]{*}{}}  & 100 \% & 100 \% & 100 \% & 100 \% & 100 \% & 100 \% \\

\multicolumn{1}{c|}{} & \multicolumn{1}{|c|}{\textit{ (VaR 1\%)}} 
& 1.17 & 1.08 & 3.92 & 0.42 & 1.00 & 2.75  & \multicolumn{1}{c|}{}  & 100 \% & 100 \% & 100 \% & 100 \% & 100 \% & 100 \% \\

\hline

\begin{sideways}\end{sideways} & \textit{} & \textbf{1.17} & \textbf{1.04} & \textbf{3.92} & \textbf{0.42} & \textbf{0.96} & \textbf{2.75}  & \begin{sideways}\end{sideways} & \textbf{100 \%} & \textbf{100 \%} & \textbf{100 \%} & \textbf{100 \%} & \textbf{100 \%} & \textbf{100 \%} \\
\hline 
\end{tabular}
		}}
\caption{\footnotesize{Time evolution of the average number of violations and Kupiec test under the Historical distribution assumption.
The table shows the evolution over the global financial crisis of the average number of violations 
and the percentage of Kupiec acceptance, aggregated at the level of $1 \% VaR$, 
as well as the increasing and decreasing $\Lambda VaR$ models.}}
\label{Table1}
\end{table}
\end{center}
As expected and already pointed out in \cite{HP15} the average number of violations of $1\%$ $VaR $ 
is bigger than the one of $\Lambda VaR$, in particular if compared with the increasing models.
In fact $1\%$ $VaR$ shows a drastic increase in the average number of violations,
moving from 3.42 in 2006 to 11.58 in 2008. On the other hand, the increasing $\Lambda VaR$ 
models register an average number of violations of around 1.17 during 2006 and retain
the number at around 3.92 in the 2008 crisis.

This result was expected since the $\Lambda$ function has been built with $\max_x \Lambda_t(x)=0.01$, 
which implies that $\Lambda VaR$ is always greater or equal than $1\%$ $VaR$, so that, losses not 
covered by the first are also not covered by the latter. This implies that $\Lambda VaR$ 
performs always better than $1 \%$ $VaR$ by using an unilateral Kupiec-type test, since 
this kind of test does not capture the variability of the $\Lambda$ function that is the essential feature of $\Lambda VaR$.

The violations trend is the same also under the other distribution's assumptions taken in exam as shown in Table \eqref{Table2}.
\begin{center}
\begin{table}[H]
 {\resizebox{1.0 \textwidth}{!}{ \renewcommand\arraystretch{2.0}  \begin{tabular}[htb] {rr|rrrrrr|rrrrrr|rrrrrr}

&\multicolumn{1}{c}{} & \multicolumn{6}{c}{\textit{\textbf{Normal }}} &    \multicolumn{6}{c}{\textit{\textbf{GARCH }}} &
\multicolumn{6}{c}{\textit{\textbf{EVT }}} \\

\hline

& &  2006 & 2007 & 2008 & 2009 & 2010 & 2011
&  2006 & 2007 & 2008 & 2009 & 2010 & 2011
&  2006 & 2007 & 2008 & 2009 & 2010 & 2011\\ 	
 \hline

\multicolumn{1}{c|}{\textit{VaR}} & \textit{1\%}
 
& 4.58  & 7.08 & 14.92 & 1.75 & 4.17 & 9.42 
& 3.17 & 6.83 & 8.25 & 0.33 & 0.75 & 4.33 
& 3.42 & 5.33 & 11.58 & 0.83 & 3.08 & 6.92
\\
\hline

& &  \textbf{4.58} & \textbf{7.08} & \textbf{14.92} &\textbf{1.75} & \textbf{4.17} & \textbf{9.42}
 & \textbf{3.17} & \textbf{6.83} & \textbf{8.25} & \textbf{0.33} & \textbf{0.75} & \textbf{4.33} 
 & \textbf{3.42} & \textbf{5.33} & \textbf{11.58} & \textbf{0.83} & \textbf{3.08} & \textbf{6.92}\\

\hline
 
\multicolumn{1}{c|}{\multirow{2}[4]{*}{\textit{$\Lambda VaR$ $1\%$ (decr)}}} & \textit{ (VaR 5\%)} 

& 4.42  & 6.75 & 14.25 & 1.58 & 3.75 & 9.17 

 & 3.08 & 5.83 & 7.33 & 0.33 & 0.42 & 4.25 
 & 2.33 & 2.33 & 7.25 & 0.75 & 1.75 & 4.25
\\
 
\multicolumn{1}{c|}{} & \textit{ (VaR 1\%)} 

 & 4.25  & 5.83 &13.08 & 1.42 & 3.42 & 8.58
  & 2.75  & 4.75 & 6.42 & 0.25 & 0.33 & 3.92 
  & 2.08 & 2.08 & 6.92 & 0.75 & 1.67 & 4.08
\\
 \hline
& & \textbf{4.33} & \textbf{6.29} & \textbf{13.67} & \textbf{1.50} & \textbf{3.58} & \textbf{8.88}
& \textbf{2.92} & \textbf{5.29} & \textbf{6.88} & \textbf{0.29} & \textbf{0.38} & \textbf{4.08}
 & \textbf{2.21} & \textbf{2.21} & \textbf{7.08} & \textbf{0.75} & \textbf{1.71} & \textbf{4.17}\\
 \hline
 
\multicolumn{1}{c|}{\multirow{2}[4]{*}{\textit{$\Lambda VaR$ $1\%$ (incr)}}} & \textit{ (VaR 5\%)}  
 & 3.33 & 4.75 & 10.83 & 0.92 & 2.75 & 6.67 
  & 1.25  & 2.67 & 3.58 & 0.00 & 0.17 & 1.42 
  & 1.25 & 1.00 & 4.25 & 0.42 & 0.92 & 2.75
\\
\multicolumn{1}{c|}{} & \textit{ (VaR 1\%)} 
  & 3.33 & 5.08 & 11.67 & 1.17 & 3.00 & 7.00
  & 1.25 & 2.83 & 3.50 & 0.00 & 0.33 & 1.42 
  & 1.25 & 1.25 & 4.33 & 0.42 & 1.00 & 2.75
\\
 \hline
&  & \textbf{3.33} & \textbf{4.92} & \textbf{11.25} & \textbf{1.04} & \textbf{2.88} & \textbf{6.83}
 & \textbf{1.25} & \textbf{2.75} & \textbf{3.54} & \textbf{0.00} & \textbf{0.25} & \textbf{1.42}
  & \textbf{1.25} & \textbf{1.13} & \textbf{4.29} & \textbf{0.42} & \textbf{0.96} & \textbf{2.75} \\
\hline 
\end{tabular}
		}}
\caption{\footnotesize{Time evolution of the average number of violations under the Normal, GARCH and EVT model.
The table shows the evolution over the global financial crisis of the average number of violations aggregated at the level of $1 \% VaR$, 
as well as the increasing and decreasing $\Lambda VaR$ models.}}
\label{Table2}
\end{table}
\end{center}

\subsubsection{Test 1 and Test 2: comparison of $VaR$ and $\Lambda VaR$ risk coverage}

In Table \eqref{Table3} and \eqref{Table4} we show the results of Test 1 and 2 proposed in 
Section \eqref{backtest} for $\Lambda VaR$.
The results here presented are under different assumptions of the distribution of the assets return, specifically, Historical, Normal, GARCH and EVT method.

\begin{center}
\begin{table}[H]
 {\resizebox{1.0 \textwidth}{!}{ \renewcommand\arraystretch{2.0}  \begin{tabular}[htb] {r r|rrrrrr|rrrrrr}

  & \multicolumn{1}{c}{} & \multicolumn{6}{c}{\textit{\textbf{Historical}}} &   \multicolumn{6}{c}{\textit{\textbf{Normal }}}  \\

\hline
& &  2006 & 2007 & 2008 & 2009 & 2010 & 2011
&  2006 & 2007 & 2008 & 2009 & 2010 & 2011\\

 \hline
 \multicolumn{1}{c|}{\textit{VaR}} & \textit{1\%} 
 & 100\% & 58\% & 0\% & 100\% & 75\% & 25\% 
 & 58\%  & 33\% & 0\% & 92\% & 50\% & 8\%\\
 \hline
 
 & & \textbf{100\%} & \textbf{58\%} & \textbf{0\%} & \textbf{100\%} & \textbf{75\%} & \textbf{25\%}
&  \textbf{58\%} & \textbf{33\%} & \textbf{0\%} &\textbf{92\%} & \textbf{50\%} & \textbf{8\%} \\
 \hline
 
\multicolumn{1}{c|}{\multirow{2}[4]{*}{\textit{$\Lambda VaR$ $1\%$ (decr)}}} 
 & \textit{ (VaR 5\%)}
 & 100 \% & 75 \% & 17 \% & 100 \% & 92 \% & 67 \% 
 & 42\%  & 8\% & 0\% & 83\% & 50\% & 8\%\\
\multicolumn{1}{c|}{} & \textit{ (VaR 1\%)} 
& 92 \% & 83 \% & 25 \% & 100 \% & 100 \% & 75 \%  
  & 33\%  & 25\% & 0\% & 92\% & 42\% & 8\% \\
 \hline

 & \textit{} & \textbf{96\%} & \textbf{79\%} & \textbf{21\%} & \textbf{100\%} & \textbf{96\%} & \textbf{71\%}
&  \textbf{38\%} & \textbf{17\%} & \textbf{0\%} &\textbf{88\%} & \textbf{46\%} & \textbf{8\%}\\
 \hline
 \multicolumn{1}{c|}{\multirow{2}[4]{*}{\textit{$\Lambda VaR$ $1\%$ (incr)}}} 
 & \multicolumn{1}{c|}{\textit{ (VaR 5\%)}}
 & 75 \% & 83 \% & 0 \% & 100 \% & 83 \% & 25 \%  
 & 0 \% & 0 \% & 0 \% & 42 \% & 33 \% & 8 \% \\
 \multicolumn{1}{c|}{} & \multicolumn{1}{|c|}{\textit{ (VaR 1\%)}} 
  & 75 \% & 83 \% & 0 \% & 100 \% & 75 \% & 17 \%  
 & 8 \% & 8 \% & 0 \% & 42 \% & 42 \% & 8 \%   \\
 
\hline

& \textit{} &   \textbf{75\% }& \textbf{83\%} & \textbf{0\%} & \textbf{100\%} & \textbf{79\%} & \textbf{21\%} 
 & \textbf{4\%} & \textbf{4\%} & \textbf{0\%} & \textbf{42\%} & \textbf{38\%} & \textbf{8\%}\\
 \hline

 & \multicolumn{1}{c}{} & \multicolumn{6}{c}{\textit{\textbf{GARCH}}} &   \multicolumn{6}{c}{\textit{\textbf{EVT }}}  \\

 \hline
 \multicolumn{1}{c|}{\textit{VaR}} & \textit{1\%} & 75\% & 50\% & 33\% & 100\% & 100\% & 67\% 
 & 100\%  & 58\% & 0\% & 100\% & 75\% & 25\%\\
 \hline
 
 & & \textbf{75\%} & \textbf{50\%} & \textbf{33\%} & \textbf{100\%} & \textbf{100\%} & \textbf{67\%}
&  \textbf{100\%} & \textbf{58\%} & \textbf{0\%} &\textbf{100\%} & \textbf{75\%} & \textbf{25\%} \\
 \hline
 
\multicolumn{1}{c|}{\multirow{2}[4]{*}{\textit{$\Lambda VaR$ $1\%$ (decr)}}} 
 & \textit{ (VaR 5\%)}
 & 75\% & 50\% & 33\% & 100\% & 100\% & 67\% 
 & 100 \% & 92 \% & 8 \% & 100 \% & 92 \% & 50 \%
\\
\multicolumn{1}{c|}{} & \textit{ (VaR 1\%)} 
& 67\%  & 67\% & 33\% & 100\% & 100\% & 67\% 
  & 100 \% & 92\% & 8 \% & 100 \% & 100 \% & 58 \% 
 \\
 \hline

 & \textit{} & \textbf{71\%} & \textbf{58\%} & \textbf{33\%} & \textbf{100\%} & \textbf{100\%} & \textbf{67\%}
&  \textbf{100\%} & \textbf{92\%} & \textbf{8\%} &\textbf{100\%} & \textbf{96\%} & \textbf{54\%}\\
 \hline
 \multicolumn{1}{c|}{\multirow{2}[4]{*}{\textit{$\Lambda VaR$ $1\%$ (incr)}}} 
 & \multicolumn{1}{c|}{\textit{ (VaR 5\%)}}
 & 67\% & 58\% & 25\% & 100\% & 92\% & 58\% 
& 67 \% & 83 \% & 0 \% & 100 \% & 83 \% & 17 \%
 \\
 \multicolumn{1}{c|}{} & \multicolumn{1}{|c|}{\textit{ (VaR 1\%)}}
 & 75\%  & 50\% & 25\% & 100\% & 92\% & 58\% 
 & 67 \% & 67 \% & 0 \% & 100 \% & 75 \% & 25 \%
 \\
 
\hline

& \textit{} &   \textbf{71\% }& \textbf{54\%} & \textbf{25\%} & \textbf{100\%} & \textbf{92\%} & \textbf{58\%} 
 & \textbf{67\%} & \textbf{75\%} & \textbf{0\%} & \textbf{100\%} & \textbf{79\%} & \textbf{21\%}\\
 \hline

\end{tabular}
}}
\caption{\footnotesize{Time evolutions of Test 1 for the $\Lambda VaR$ models under 
different assumptions of the P$\&$L distribution. The table shows the evolution over the 
global financial crisis of the acceptance rates, aggregated at the level of the $\Lambda 
VaR$ models ($\min_x\Lambda(x)=0.5 \%$) calculated using the Historical, Normal, GARCH and EVT 
assumption of the P$\&$L distribution.}}
\label{Table3}
\end{table}
\end{center}

\begin{center}
\begin{table}[H]
 {\resizebox{1.0 \textwidth}{!}{ \renewcommand\arraystretch{2.0}  \begin{tabular}[htb] {r r | rrrrrr|rrrrrr}

 & \multicolumn{1}{c}{} & \multicolumn{6}{c}{\textit{\textbf{Historical}}} &   \multicolumn{6}{c}{\textit{\textbf{Normal }}} \\

\hline
& &  2006 & 2007 & 2008 & 2009 & 2010 & 2011
&  2006 & 2007 & 2008 & 2009 & 2010 & 2011\\

 \hline
 \multicolumn{1}{c|}{\textit{VaR}} & \multicolumn{1}{c|}{\textit{1\%} }
&100 \% & 75 \% & 0 \% & 100 \% & 92 \% & 42 \% 
& 58\%  & 42\% & 0\% & 100\% & 67\% & 25\%\\
 \hline
 
 & & \textbf{100\%} & \textbf{75\%} & \textbf{0\%} & \textbf{100\%} & \textbf{92\%} & \textbf{42\%}
&  \textbf{58\%} & \textbf{42\%} & \textbf{0\%} &\textbf{100\%} & \textbf{67\%} & \textbf{25\%} \\
 \hline
 
\multicolumn{1}{c|}{\multirow{2}[4]{*}{\textit{$\Lambda VaR$ $1\%$ (decr)}}} 
 & \textit{ (VaR 5\%)}
 & 100\% & 83\% & 17\% & 100\% & 100\% & 75\% 
  & 58\%  & 42\% & 0\% & 100\% & 67\% & 17\%\\
\multicolumn{1}{c|}{} & \textit{ (VaR 1\%)} 
 & 100\% & 83\% & 42\% & 100\% & 100\% & 83\%  
  & 50\%  & 50\% & 0\% & 100\% & 67\% & 17\% \\
 \hline

 & \textit{} & \textbf{100\%} & \textbf{83\%} & \textbf{29\%} & \textbf{100\%} & \textbf{100\%} & \textbf{79\%}
&  \textbf{54\%} & \textbf{46\%} & \textbf{0\%} &\textbf{100\%} & \textbf{67\%} & \textbf{17\%}\\
 \hline
 \multicolumn{1}{c|}{\multirow{2}[4]{*}{\textit{$\Lambda VaR$ $1\%$ (incr)}}} 
 & \multicolumn{1}{c|}{\textit{ (VaR 5\%)}}
 & 100\% & 100\% & 17\% & 100\% & 92\% & 42\%  
& 17\% & 25\% & 0\% & 92\% & 50\% & 8\% \\
 \multicolumn{1}{c|}{} & \multicolumn{1}{|c|}{\textit{ (VaR 1\%)}} 
 & 100\% & 100\% & 17\% & 100\% & 92\% & 42\% 
& 25\% & 33\% & 0\% & 83\% & 58\% & 25\%\\
 
\hline

& \textit{} &   \textbf{100\%} & \textbf{100\%} & \textbf{17\%} & \textbf{100\%} & \textbf{92\%} & \textbf{42\%}

 & \textbf{21\%} & \textbf{29\%} & \textbf{0\%} & \textbf{88\%} & \textbf{54\%} & \textbf{17\%}\\
 \hline

 & \multicolumn{1}{c}{} & \multicolumn{6}{c}{\textit{\textbf{GARCH}}} &   \multicolumn{6}{c}{\textit{\textbf{EVT }}}  \\

\hline
 \multicolumn{1}{c|}{\textit{VaR}} & \textit{1\%}  & 83\% & 58\% & 42\% & 100\% & 100\% & 67\% 
 &100 \% & 75 \% & 0 \% & 100 \% & 92 \% & 33 \% \\
 \hline
 
 & & \textbf{83\%} & \textbf{58\%} & \textbf{42\%} & \textbf{100\%} & \textbf{100\%} & \textbf{67\%}
&  \textbf{100\%} & \textbf{75\%} & \textbf{0\%} &\textbf{100\%} & \textbf{92\%} & \textbf{33\%} \\
 \hline
 
\multicolumn{1}{c|}{\multirow{2}[4]{*}{\textit{$\Lambda VaR$ $1\%$ (decr)}}} 
 & \textit{ (VaR 5\%)}
  & 83\% & 58\% & 33\% & 100\% & 100\% & 67\% 
&100 \% & 92 \% & 8 \% & 100 \% & 100 \% & 67 \% \\
 
\multicolumn{1}{c|}{} & \textit{ (VaR 1\%)} 
 & 92\%  & 75\% & 42\% & 100\% & 100\% & 75\% 
  & 100 \% & 92 \% & 17 \% & 100 \% & 100 \% & 67 \%  
 \\
 \hline

& & \textbf{88\%} & \textbf{67\%} & \textbf{38\%} & \textbf{100\%} & \textbf{100\%} & \textbf{71\%}
&  \textbf{100\%} & \textbf{92\%} & \textbf{13\%} &\textbf{100\%} & \textbf{100\%} & \textbf{67\%}\\
 \hline
 \multicolumn{1}{c|}{\multirow{2}[4]{*}{\textit{$\Lambda VaR$ $1\%$ (incr)}}} 
 & \multicolumn{1}{c|}{\textit{ (VaR 5\%)}}
   & 92\% & 75\% & 67\% & 100\% & 100\% & 83\%
& 100 \% & 100 \% & 17 \% & 100 \% & 92 \% & 42 \%  \\
 
 \multicolumn{1}{c|}{} & \multicolumn{1}{|c|}{\textit{ (VaR 1\%)}}
 & 92\% & 67\% & 67\% & 100\% & 92\% & 83\%  
 & 100 \% & 92 \% & 17 \% & 100 \% & 92 \% & 42 \%
 \\
 
\hline
& & \textbf{92\%} & \textbf{71\%} & \textbf{67\%} & \textbf{100\%} & \textbf{96\%} & \textbf{83\%}
 & \textbf{100\%} & \textbf{96\%} & \textbf{17\%} & \textbf{100\%} & \textbf{92\%} & \textbf{42\%}\\
 
 \hline

\end{tabular}
}}
\caption{\footnotesize{Time evolutions of Test 2 for the $\Lambda VaR$ models under 
different assumptions of the P$\&$L distribution. The table shows the evolution over the 
global financial crisis of the acceptance rates, aggregated at the level of the $\Lambda 
VaR$ models ($\min_x\Lambda(x)=0.5 \%$) calculated using the Historical, Normal, GARCH and EVT
assumption of the P$\&$L distribution.}}
\label{Table4}
\end{table}
\end{center}

We first notice that the acceptance rate of these tests is lower than the unilateral Kupiec test in \cite{HP15}. 
This is due to the particular construction of the Kupiec test. Indeed, this test is useful to assess if the 
$\Lambda VaR$ model guarantees an acceptable coverage given by $\max(\Lambda)$, but cannot capture the daily variations of the confidence level $\lambda^0_t$ of $\Lambda VaR$. Thus, it cannot be used to evaluate the real coverage offered by $\Lambda VaR$ at time $t$. On the other hand, the coverage tests that we have proposed are able to better evaluate if the flexibility 
introduced by the $\Lambda$ function helps to detect adverse scenario and put aside a more adequate amount of capital. 

If we compare the tests results, we observe that for all the models Test 2 provides higher acceptance rates in respect to Test 1. This may be due to the fact that Test 1 returns more precise results 
 with smaller number of observations and also to its unilateral nature that attributes the highest weight to the violations.   

With the exception of the normal estimator, the $\Lambda VaR$ models result often more accurate than $1\%$ $VaR$, confirming the outcomes 
in \cite{HP15}. This means that the highest flexibility of $\Lambda VaR$ contributes to the highest 
coverage, especially when it is computed with distributions that better capture the tail behaviour. In our tests, the decreasing $\Lambda VaR$ models 
seem to be more accurate, in contrast with the results of the Kupiec test. We think this is a consequence of a lower power of these tests for the decreasing $\Lambda VaR$ models. We remand the analysis of the test power for further research since it would complicate this study without adding significant value.

\subsubsection{The choice of the $\Lambda$ minimum}\label{sec:minimum}

During the analysis of the results, Test 1 and Test 2 have pointed out an issue of 
estimation in the $\Lambda VaR$ models proposed by \cite{HP15}. In particular, the authors do not discuss in details the choice of the $\Lambda$ minimum, 
$\min_x\Lambda(x)$, that seems to be set equal to $0.1\%$ after empirical experimentations. In addition, the extended Kupiec test proposed by the authors could not identify the impact of this choice.

When we have run for the first time Test 1 and 2 using the choice of \cite{HP15}, $\min_x\Lambda(x)=0.1 \%$, we have noticed that the increasing $\Lambda VaR$ models presented the highest rejection rate, even if they had the smallest number of infractions, as shown by Table \eqref{Table5}.
\begin{center}
\begin{table}[H]
 {\resizebox{1 \textwidth}{!}{ \renewcommand\arraystretch{2.0} \begin{tabular}[htb] {rr|rrrrrr r|rrrrrr}

 & \multicolumn{1}{c}{}& \multicolumn{6}{c}{\textit{\textbf{Test 1}}} &             \multicolumn{6}{c}{\textit{\textbf{Test 2}}} \\
\hline
& \multicolumn{1}{c}{\textit{}} & \multicolumn{1}{|c}{2006} & \multicolumn{1}{c}{2007} & \multicolumn{1}{c}{2008} & \multicolumn{1}{c}{2009} & \multicolumn{1}{c}{2010} & \multicolumn{1}{c}{2011} & \multicolumn{1}{c}{} & \multicolumn{1}{|c}{2006} & \multicolumn{1}{c}{2007} & \multicolumn{1}{c}{2008} & \multicolumn{1}{c}{2009} & \multicolumn{1}{c}{2010} & \multicolumn{1}{c}{2011} \\

\hline

\multicolumn{1}{c|}{\multirow{2}[4]{*}{\textit{$\Lambda VaR$ $1\%$ (decr)}}} & \multicolumn{1}{|c|}{\textit{ (VaR 5\%)}} & 100 \% & 75 \% & 8 \% & 100 \% & 92 \% & 67 \% 
& \multicolumn{1}{c|}{\multirow{2}[4]{*}{}}  & 100 \% & 83 \% & 17 \% & 100 \% & 100 \% & 75 \% \\

\multicolumn{1}{c|}{} & \multicolumn{1}{|c|}{\textit{ (VaR 1\%)}} & 92 \% & 83 \% & 25 \% & 100 \% & 100 \% & 67 \%
&  \multicolumn{1}{c|}{}  & 100 \% & 83 \% & 42 \% & 100 \% & 100 \% & 83 \% \\
\hline

\begin{sideways}\end{sideways} & \textit{} &   \textbf{96 \% }& \textbf{79 \%} & \textbf{17 \%} & \textbf{100 \%} & \textbf{96 \%} & \textbf{67 \%} & 
 \multicolumn{1}{c|}{\begin{sideways}\end{sideways}}& \textbf{100 \%} & \textbf{83 \%} & \textbf{29 \%} & \textbf{100 \%} & \textbf{100 \%} & \textbf{79 \%}\\
\hline

\multicolumn{1}{c|}{\multirow{2}[4]{*}{\textit{$\Lambda VaR$ $1\%$ (incr)}}} & \multicolumn{1}{|c|}{\textit{(VaR 5\%)}} &  8 \% &  17 \% & 0 \% & 58 \% & 42 \% & 8 \% 
&  \multicolumn{1}{c|}{\multirow{2}[4]{*}{}}  & 75 \% & 83 \% & 0 \% & 100 \% & 75 \% & 17 \%  \\

\multicolumn{1}{c|}{} & \multicolumn{1}{|c|}{\textit{ (VaR 1\%)}} &  8 \% & 17 \% & 0 \% & 58 \% & 42 \% & 8 \% 
&  \multicolumn{1}{c|}{}  & 75 \% & 83 \% & 0 \% & 100 \% & 83 \% & 25 \% \\
\hline

\begin{sideways}\end{sideways} & \textit{} &   \textbf{8 \%} & \textbf{17 \%} & \textbf{0 \%} & \textbf{58 \%} & \textbf{42 \%} & \textbf{8 \%} 
&\begin{sideways}\end{sideways} & \textbf{75 \%} & \textbf{83 \%} & \textbf{0 \%} & \textbf{100 \%} & \textbf{79 \%} & \textbf{21 \%} \\
\hline 
\end{tabular}
}}
\caption{\footnotesize{Time evolutions of Test 1 and Test 2 for the $\Lambda VaR$ models 
with $\min_x\Lambda(x)=0.1 \%$ under the Historical distribution assumption. The table shows the evolution over the global financial crisis 
of the acceptance rates, aggregated at the level of the $\Lambda VaR$ models with $\min_x\Lambda(x)=0.1 \%$.}}
\label{Table5}
\end{table}
\end{center}

Thus, we have studied how the probability of infraction $\lambda_t$ evolves in the different $\Lambda VaR$ models and we have 
observed that in most of the cases it obtains the minimal value. This happens 
especially during crisis periods, when the cumulative distribution function of the assets 
shifts on the left and intersects the $\Lambda$ function at the minimum level. In such a case, the choice of the $\Lambda$ minimum is relevant and also a critical issue.

From our point of view, the $\Lambda$ minimum should 
provide the probability to lose more than the \textit{worst case event} (i.e. benchmarks' 
minimum, $\pi_{1}=\min x_{t,j}$) over the time window observations (i.e. 250 in our case). 
If we consider all the events equally probable, the selection of the $\Lambda$ minimum 
should be greater than $1/T$ over $T$ observations. Thus, we propose to compute the $\Lambda VaR$ models by fixing the $\Lambda$ minimum 
equal to $0.5\%$, i.e. $\min_x\Lambda(x)=0.005$, since the probability of an event over 250 past realizations is $0.4\%$. 

The results of the $\Lambda VaR$ estimations with $0.5\%$ minimum have been shown before in Table \eqref{Table3} and \eqref{Table4}. The number of infractions does not 
change in any period under consideration,
while the acceptance rate of the increasing $\Lambda VaR$ models drastically increases, validating our choice. Clearly, this new setting does not affect 
the decreasing $\Lambda VaR$ models.
Anyway, the choice of the $\Lambda$ minimum can be refined considering more precise 
evaluation of the probability of the worst case event, but this is beyond the objective 
of this paper.

\subsubsection{Test 3: comparison of $\Lambda VaRs$ with different distribution estimations}
As anticipated in Section \eqref{backtest}, the best use of Test 3 is the 
comparison of the accuracy of the risk measures computed with different estimations of asset return distribution.  We 
compute the time evolution of the acceptance rate aggregated at the level of the increasing 
and decreasing $\Lambda VaR$ models. We repeat the analysis changing the assumption on the 
asset return distribution: specifically, Historical, Monte Carlo Normal, GARCH and EVT method. The 
results are presented in Table \eqref{Table6}

\begin{center}
\begin{table}[H]
 {\resizebox{1.0 \textwidth}{!}{ \renewcommand\arraystretch{2.0}  \begin{tabular}[htb] {r r|rrrrrr|rrrrrr}

  & \multicolumn{1}{c}{} & \multicolumn{6}{c}{\textit{\textbf{Historical}}} &   \multicolumn{6}{c}{\textit{\textbf{Normal }}}  \\

\hline
& &  2006 & 2007 & 2008 & 2009 & 2010 & 2011
&  2006 & 2007 & 2008 & 2009 & 2010 & 2011\\

 \hline
 \multicolumn{1}{c|}{\textit{VaR}} & \textit{1\%} 
 & 50\% & 33\% & 0\% & 100\% & 58\% & 25\% 
 
 & 58\%  & 33\% & 0\% & 92\% & 50\% & 8\%\\
 \hline
 
& & \textbf{50\%} & \textbf{33\%} & \textbf{0\%} & \textbf{100\%} & \textbf{58\%} & \textbf{25\%}
&  \textbf{58\%} & \textbf{33\%} & \textbf{0\%} &\textbf{92\%} & \textbf{50\%} & \textbf{8\%} \\
 
 \hline
 
\multicolumn{1}{c|}{\multirow{2}[4]{*}{\textit{$\Lambda VaR$ $1\%$ (decr)}}} 
 & \textit{ (VaR 5\%)}& 50\% & 33\% & 0\% & 100\% & 67\% & 17\% 
 & 58\%  & 42\% & 0\% & 92\% & 58\% & 17\%\\
 
\multicolumn{1}{c|}{} & \textit{ (VaR 1\%)} & 58\% & 50\% & 8\% & 100\% & 67\% & 8\% 
  & 50\%  & 33\% & 0\% & 92\% & 58\% & 25\% \\
 \hline

&  & \textbf{54\%} & \textbf{42\%} & \textbf{4\%} & \textbf{100\%} & \textbf{67\%} & \textbf{13\%}
&  \textbf{54\%} & \textbf{38\%} & \textbf{0\%} & \textbf{92\%} & \textbf{58\%} & \textbf{21\%}\\
 
 \hline
 \multicolumn{1}{c|}{\multirow{2}[4]{*}{\textit{$\Lambda VaR$ $1\%$ (incr)}}} 
 & \multicolumn{1}{c|}{\textit{ (VaR 5\%)}}& 8\% & 17\% & 0\% & 58\% & 42\% & 0\% 
 & 17\% & 17\% & 0\% & 92\% & 50\% & 8\% \\
 \multicolumn{1}{c|}{} & \multicolumn{1}{|c|}{\textit{ (VaR 1\%)}} & 8\% & 17\% & 0\% & 58\% & 42\% & 8\% 
 & 33\% & 8\% & 0\% & 83\% & 50\% & 17\% \\
 
\hline

& \textit{} & \textbf{8\%} & \textbf{17\%} & \textbf{0\%} & \textbf{58\%} & \textbf{42\%} & \textbf{4\%} 
 & \textbf{25\%} & \textbf{13\%} & \textbf{0\%} & \textbf{88\%} & \textbf{50\%} & \textbf{13\%}\\
 \hline

 & \multicolumn{1}{c}{} & \multicolumn{6}{c}{\textit{\textbf{GARCH}}} &   \multicolumn{6}{c}{\textit{\textbf{EVT }}}  \\

 \hline
 \multicolumn{1}{c|}{\textit{VaR}} & \textit{1\%} 
 & 75\% & 58\% & 33\% & 100\% & 100\% & 67\% 
 & 50 \% & 33 \% & 0 \% & 100 \% & 58 \% & 25 \% \\
 \hline
 
 & & \textbf{75\%} & \textbf{58\%} & \textbf{33\%} & \textbf{100\%} & \textbf{100\%} & \textbf{67\%}
&  \textbf{50\%} & \textbf{33\%} & \textbf{0\%} &\textbf{100\%} & \textbf{58\%} & \textbf{25\%} \\
 \hline
 
\multicolumn{1}{c|}{\multirow{2}[4]{*}{\textit{$\Lambda VaR$ $1\%$ (decr)}}} 
 & \textit{ (VaR 5\%)}
 & 75\% & 58\% & 33\% & 100\% & 100\% & 67\% 
 & 67 \% & 58 \% & 0 \% & 92 \% & 58 \% & 42 \% 
\\
 
\multicolumn{1}{c|}{} & \textit{ (VaR 1\%)} 
& 92\%  & 67\% & 33\% & 100\% & 100\% & 75\% 
& 67 \% & 58 \% & 0 \% & 92 \% & 58 \% & 33 \%
 \\
 \hline

 & & \textbf{83\%} & \textbf{63\%} & \textbf{33\%} & \textbf{100\%} & \textbf{100\%} & \textbf{71\%}
&  \textbf{67\%} & \textbf{58\%} & \textbf{0\%} &\textbf{92\%} & \textbf{58\%} & \textbf{38\%}\\
 \hline
 
 \multicolumn{1}{c|}{\multirow{2}[4]{*}{\textit{$\Lambda VaR$ $1\%$ (incr)}}} 
 & \multicolumn{1}{c|}{\textit{ (VaR 5\%)}}
 & 83\% & 67\% & 67\% & 100\% & 100\% & 83\% 
& 17 \% & 25 \% & 8 \% & 58 \% & 42 \% & 0 \% 
 \\
 \multicolumn{1}{c|}{} & \multicolumn{1}{|c|}{\textit{ (VaR 1\%)}}
  & 83\% & 58\% & 67\% & 100\% & 92\% & 83\%
& 17 \% & 33 \% & 0 \% & 58 \% & 42 \% & 8 \%  \\
 
\hline

& & \textbf{83\%} & \textbf{63\%} & \textbf{67\%} & \textbf{100\%} & \textbf{96\%} & \textbf{83\%} 
 & \textbf{17\%} & \textbf{29\%} & \textbf{4\%} & \textbf{58\%} & \textbf{42\%} & \textbf{4\%}\\
 \hline

\end{tabular}
}}
\caption{\footnotesize{Time evolutions of Test 3 for the $\Lambda VaR$ models under 
different assumptions of the P$\&$L distribution. The table shows the evolution over the 
global financial crisis of the acceptance rates, aggregated at the level of the $\Lambda 
VaR$ models ($\min_x\Lambda(x)=0.5 \%$) calculated using the Historical, Normal, GARCH and EVT 
assumption of the P$\&$L distribution.}}
\label{Table6}
\end{table}
\end{center}

The results show that the GARCH assumption on the returns guarantees the highest accuracy in terms of average acceptance rate. Moreover, we notice here that the Historical and the EVT estimators of the increasing $\Lambda VaR$ often underperform the Normal one, in contrast with the previous tests. These outcomes are quite reasonable since this third test is based on simulations and points out the issue of estimating risk measures with distributions having cut-off tails (as the Historical) or based on a small range of values (as the EVT). 
However, such a preference for the Normal distribution is completely reversed by the other tests which privilege the assumption of distributions which rely more on tail events and not on the full shape of the distribution.

\section{Conclusions}\label{sec:conclusioni}

A new risk measure sensitive to tail risk, $\Lambda VaR$, has been recently introduced. However, an \textit{ad hoc} study on its 
backtesting has not been conducted in literature so far. The main issue for the $\Lambda VaR$ backtesting is that the probability of a violation is not constant, but may change at any time and for any asset. This consideration implies that the Kupiec-type backtesting framework, proposed by \cite{HP15}, fails to keep into account the effective 
predictive capacity of $\Lambda VaR$ as introduced by the $\Lambda$ function.

We propose three backtesting methodologies for $\Lambda VaR$ and we asses the accuracy of the new risk measure from different points of view. Test 1 and Test 2 are based on results of probability theory and allow for a straightforward application. Test 3 is performed by simulations and allows for more accurate comparison of $\Lambda VaR$ models estimated under different assumptions on the P$\&$L distribution.

The validity of our backtesting proposals is confirmed by the results of the empirical analysis. In fact, this study shows that $\Lambda VaR$ models perform better 
than $1$ $\% $ $VaR$, confirming the findings in \cite{HP15}. In addition, $\Lambda VaR$ computed with the GARCH model of returns has the highest level of coverage. This outcome substantiates what is well known in literature that fat-tailed asset return distributions explain better the real asset return behavior and allow for a more accurate risk coverage.

Moreover, our backtesting methods denote higher precision than the Kupiec-type test proposed by \cite{HP15} since they have been able to detect an estimation issue of $\Lambda VaR$ computed with a lower bound of $0.1\%$ as in the former study. 

Suggestions for future research include the study of the test power that would permit a more accurate comparison among these backtesting proposals.

\section*{Appendix} \label{appendix}
We recall hereafter the Lyapunov Theorem that is a result of probability theory based on the application of the central limit theorem to random 
variables that are independent but not identically distributed (see \citealt{L54}).

\begin{theorem}[Lyapunov]
Suppose $X_1, \, X_2, ...$ is a sequence of independent random variables, each 
with finite expected value $\mu_t$ and variance $\sigma^2_t$. Define
\[
s_n^2=\sum_{t=1}^T\sigma^2_t
\]
If for some $\delta > 0$, the ``Lyapunov's condition'' 
\[
\lim_{n\to\infty} \frac{1}{s_{T}^{2+\delta}} \sum_{t=1}^{T} \E\big[\,|X_{t} - \mu_{t}|^{2+\delta}\,\big] = 0
\]
is satisfied, then the following convergence in distribution holds as $T$ goes to 
infinity:
\[
\frac{1}{s_T} \sum_{t=1}^{T} (X_t - \mu_t) \ \xrightarrow{d}\ \mathcal{N}(0,\;1)
\]
\end{theorem}

In the following lemma we show that the ``Lyapunov's condition'' is satisfied when 
$s_T^2=\sum_1^T \lambda_t(1-\lambda_t)$ and $\mu_t=\lambda_t$. 
\begin{lemma}
If $\{I_t\}$ is a sequence of independent random variables distributed as a 
Bernoulli with parameters $\{ \lambda_t \}_t$ and $\inf_{t} \lambda_t=\lambda_m>0$, then
\[
 \lim_{T\to\infty} \frac{1}{s_T^{2+\delta}}\sum_{t=1}^T\E[|I_t-\lambda_t|^{2+ \delta}]=0
\]
with $s_T^2=\sum_1^T \lambda_t(1-\lambda_t)$.
\label{LemmaLyap}
\end{lemma}
\begin{proof}
We observe that:
\begin{displaymath}
\begin{split}
\E[|I_t-\lambda_t|^{2+\delta}]&=(1-\lambda_t)\lambda_t^{2+\delta}+\lambda_t(1-\lambda_t)^{2+\delta}\\
&=\lambda_t(1-\lambda_t)\left(\lambda_t^{1+\delta}+(1-\lambda_t)^{1+\delta}\right)\leq \lambda_t(1-\lambda_t)\leq \frac14\,.
\end{split}
\end{displaymath} 
On the other hand we have
\[
s_T^{2+\delta}=\left(\sum_1^T \lambda_t(1-\lambda_t)\right)^{1+\frac{\delta}{2}}\geq
 \left(\sum_1^T \lambda_m(1-\lambda_m)\right)^{1+\frac{\delta}{2}}=
\left(T \lambda_m(1-\lambda_m)\right)^{1+\frac{\delta}{2}}\,.
\]
We can thus conclude that
\begin{displaymath}
\frac{\sum_{t=1}^T\E[|I_t-\lambda_t|^{2+\delta}]}{s_T^{2+\delta}} \leq 
\frac{T}{4\left(T \lambda_m(1-\lambda_m)\right)^{1+\frac{\delta}{2}}}\to 0
\end{displaymath}
as $T\to\infty $.
\end{proof}

The following proposition shows theoretical implications on the $Z_3$ test statistic in \eqref{Z3} under the null and alternative hypothesis.

\begin{proposition}\label{prop}
Under the test hypothesis $H_0$ as in \eqref{null} and $H_1$ as in \eqref{alternative2} we have:
\begin{enumerate}
\item $\E_{H_0}[Z_3]=0$
\item $\E_{H_1}[Z_3]<0$. 
\end{enumerate}
\end{proposition}
\begin{proof}
It is enough to notice that under $H_0$, $I_t\sim B(\lambda^0_t)$ so that $\E_{H_0}[I_t-\lambda^0_t]=0$, 
which implies $$\E_{H_0}[Z_3]=\frac1T \sum \E_{H_0}[\lambda^0_t-I_t]=0\,. $$
In a similar way, under $H_1$, since $I_t\sim B(\lambda_t)$ with $\lambda_t > \lambda^0_t$, we obtain that $\E_{H_1}[Z_3]<0$.

\end{proof}

\end{document}